\let\TeXyear\year

\documentclass{ieeeaccess}

\let\year\TeXyear

\usepackage[T1]{fontenc}
\usepackage[utf8]{inputenc}
\usepackage{textcomp}
\usepackage{microtype}
\usepackage{cite}

\usepackage{amsmath,amssymb,amsfonts,amsthm}
\usepackage{mathtools}
\usepackage{bm}

\mathtoolsset{mathic=true}

\usepackage{array}
\usepackage{booktabs}
\usepackage{multirow}

\usepackage{graphicx}
\usepackage{placeins} 

\makeatletter
\providecommand{\xfigwd}{\columnwidth}

\AddToHook{env/figure/begin}{%
  \gdef\xfigwd{\columnwidth}%
}

\AddToHook{env/figure*/begin}{%
  \gdef\xfigwd{\textwidth}%
}
\makeatother

\newcounter{tqesubfigure}
\renewcommand{\thetqesubfigure}{\alph{tqesubfigure}}

\AddToHook{env/figure/begin}{%
  \setcounter{tqesubfigure}{0}%
}

\AddToHook{env/figure*/begin}{%
  \setcounter{tqesubfigure}{0}%
}

\newcommand{\subfloat}[2][]{%
  \refstepcounter{tqesubfigure}%
  \begin{minipage}{\linewidth}
    \centering
    #2\par
    \vspace{1pt}%
    {\footnotesize
      (\thetqesubfigure)%
      \if\relax\detokenize{#1}\relax
      \else
        \enspace #1%
      \fi
    }%
  \end{minipage}%
}

\usepackage{tikz}

\usetikzlibrary{
  arrows.meta,
  positioning,
  shapes.geometric,
  backgrounds,
  fit
}

\definecolor{accessblue}{cmyk}{1,0.3,0,0.2}
\definecolor{greycolor}{cmyk}{0,0,0,0.8}
\definecolor{grey}{cmyk}{0,0,0,0.1}
\definecolor{black}{cmyk}{0,0,0,1}

\usepackage[hidelinks,bookmarks=false]{hyperref}

\usepackage{balance}

\makeatletter

\def\historyfont{%
  \sffamily\mdseries\upshape
  \fontsize{7}{8.4}\selectfont
}

\def\doifont{%
  \rmfamily\mdseries\itshape
  \fontsize{6}{7.2}\selectfont
}

\def\titlefont{%
  \sffamily\bfseries\upshape
  \fontsize{22}{25.4}\selectfont
}

\def\authorfont{%
  \sffamily\bfseries\upshape
  \fontsize{9.9}{12}\selectfont
}

\def\membershipfont{%
  \sffamily\bfseries\itshape
  \fontsize{9.9}{12}\selectfont
}

\def\addressfont{%
  \rmfamily\mdseries\upshape
  \fontsize{6.6}{8}\selectfont
}

\def\correspfont{%
  \rmfamily\mdseries\upshape
  \fontsize{7.7}{9.24}\selectfont
}

\def\tfootnotefont{%
  \rmfamily\mdseries\upshape
  \fontsize{7.7}{9.24}\selectfont
}

\def\abstractfont{%
  \rmfamily\mdseries\upshape
  \fontsize{10}{12}\selectfont
}

\def\bodyfont{%
  \rmfamily\mdseries\upshape
  \fontsize{10}{12}\selectfont
}

\def\abstractheadfont{%
  \sffamily\bfseries\upshape
  \fontsize{10}{12}\selectfont
}

\def\indexheadfont{%
  \sffamily\bfseries\upshape
  \fontsize{10}{12}\selectfont
}

\def\sectionAfont{%
  \sffamily\bfseries\upshape
  \fontsize{9}{12}\selectfont
}

\def\sectionBfont{%
  \sffamily\bfseries\itshape
  \fontsize{9}{12}\selectfont
}

\def\sectionCfont{%
  \sffamily\mdseries\upshape
  \fontsize{9}{12}\selectfont
}

\def\sectionDfont{%
  \sffamily\mdseries\itshape
  \fontsize{9}{12}\selectfont
}

\def\headerfont{%
  \sffamily\mdseries\upshape
  \fontsize{7}{7}\selectfont
}

\def\iheaderfont{%
  \sffamily\mdseries\itshape
  \fontsize{7}{7}\selectfont
}

\renewcommand{\section}{%
  \@startsection
    {section}
    {1}
    {1pt}
    {-1pc plus -1pt minus -1pt}
    {0.001pt}
    {\raggedright\color{accessblue}\sectionAfont\MakeUppercase}%
}

\renewcommand{\subsection}{%
  \@startsection
    {subsection}
    {2}
    {1pt}
    {-1pc plus -1pt minus -1pt}
    {1pt}
    {\raggedright\color{greycolor}\sectionBfont\MakeUppercase}%
}

\makeatother

\newtheorem{theorem}{Theorem}
\newtheorem{proposition}{Proposition}

\theoremstyle{definition}

\theoremstyle{remark}

\def\BibTeX{{\rm B\kern-.05em{\sc i\kern-.025em b}\kern-.08em
  T\kern-.1667em\lower.7ex\hbox{E}\kern-.125emX}}

\history{}
\doi{}
\title{Adaptive Microwave Reception with Diamond Spin Sensors}
\author{\uppercase{Amirhossein Taherpour}\authorrefmark{1},
\uppercase{Abbas Taherpour}\authorrefmark{2},
\uppercase{Mazen Hasna}\authorrefmark{2},
\uppercase{Tamer Khattab}\authorrefmark{2}, and
\uppercase{Khalid Qaraqe}\authorrefmark{3}}
\address[1]{Department of Electrical Engineering, Columbia University, New York, NY 10027, USA}
\address[2]{Intelligent Information Processing Lab (IIPL), Department of Electrical Engineering, Qatar University, Doha, Qatar}
\address[3]{Division of Computing and Electrical Engineering, College of Science and Engineering, Hamad Bin Khalifa University, Doha, Qatar}
\corresp{Corresponding author: Amirhossein Taherpour (e-mail: at3532@columbia.edu).}
\AtBeginDocument{%
 \setlength{\textfloatsep}{5pt plus 2pt minus 2pt}%
 \setlength{\floatsep}{5pt plus 2pt minus 2pt}%
 \setlength{\intextsep}{5pt plus 2pt minus 2pt}%
 \setlength{\abovedisplayskip}{5pt plus 2pt minus 2pt}%
 \setlength{\belowdisplayskip}{5pt plus 2pt minus 2pt}%
}
\begin{document}
\begin{abstract}
In a diamond-spin microwave receiver, pilot measurements can identify the spin-control phase that best separates binary phase-shift keyed symbols. Applying that phase takes time, during which carrier drift can erase the gain in alignment. We derive an exact feedback-delay threshold for a single payload spin from the phase information in the pilot record. On a finite menu of control phases, the pilot record enters the minimum-error update decision only through the direction of its first circular moment when that moment is nonzero. This result motivates a controller that estimates the direction from orthogonal pilot measurements and retains the existing phase whenever the predicted gain is lost to diffusion. For finite spin ensembles, we bound the distinguishability of the photon-count distributions in terms of optical readout and carrier-phase uncertainty, including fluctuations shared across spins. Packet simulations at matched spin budgets show that gated feedback can improve reliability with a single readout channel. With independently readable quadratures and sufficient spin preparations, decision-directed software tracking achieves lower packet error without a hardware phase update.
\end{abstract}
\titlepgskip=-15pt
\maketitle

\section{Introduction}\label{sec:intro}
\IEEEPARstart{D}{iamond} nitrogen-vacancy (NV) centers respond to microwave magnetic fields through changes in their spin state, which can be detected optically. Their long-lived spin coherence and room-temperature operation make them attractive for specialized microwave sensing, particularly when direct access to the local magnetic field is useful~\cite{Doherty2013,Degen2017,Barry2020}. Experiments have established coherent radio-frequency detection with solid spins, digitally modulated reception with nanodiamonds, and high-bandwidth heterodyne sensing with bulk diamond~\cite{SolidSpins2023,Zeng2025,Barson2026}. In contrast to a conventional coherent receiver, however, a spin sensor interrogates the field over a finite interval and returns noisy fluorescence counts after preparation and readout. Its sensing response depends on both the field phase and the phase of the applied spin-control sequence.

\begin{figure*}[!t]
\centering
\begin{tikzpicture}[font=\scriptsize,>=Latex,node distance=3mm and 2.4mm]
\tikzstyle{cell}=[draw=black!55,rounded corners=2pt,fill=blue!4,align=center,minimum height=9mm,text width=26mm]
\node[cell](rf){Microwave source\\known pilots, then\\unknown BPSK bits};
\node[cell,right=of rf](nv){Controlled NV spins\\prepare, interrogate\\and read out};
\node[cell,right=of nv](opt){Photon counter\\fluorescence counts\\$C_1,C_2,C_d$};
\node[cell,right=of opt](est){Receiver processing\\pilot I/Q or\\posterior reference};
\node[cell,right=of est](out){Data decisions\\BER, packet error\\and elapsed time};
\draw[->](rf)--(nv);\draw[->](nv)--(opt);\draw[->](opt)--(est);\draw[->](est)--(out);
\node[draw=black!50,rounded corners=2pt,fill=black!3,below=8mm of nv,align=center,text width=46mm](control){Microwave control phase $\alpha$\\spin preparation and readout schedule};
\node[draw=black!50,rounded corners=2pt,fill=black!3,below=8mm of est,align=center,text width=47mm](switch){Feedback decision\\use updated phase or skip reprogramming};
\draw[->](control.north)--(nv.south);\draw[->](est.south)--(switch.north);\draw[->](switch.west)--(control.east);
\end{tikzpicture}
\caption{Diamond-spin receiver. Pilot counts establish the carrier reference. The lower path updates the spin-control phase.}
\label{fig:architecture}
\end{figure*}

Known pilots supply a carrier reference, but each pilot consumes spins that could be allocated to data. A receiver that acts on this reference must also allow time for its control phase to settle. Parallel phase-diverse subensembles avoid this delay if each quadrature can be addressed and read independently.

Adaptive phase measurements use previous outcomes to choose a measurement axis for a fluctuating optical phase~\cite{BerryWiseman2002,BerryWiseman2004}. Cost-aware sensing and learned diamond-spin control account for measurement expense and device variation~\cite{McMichael2025,Jauch2026}. For phase-modulated communications, Watkins \emph{et al.} analyze bit error and channel capacity for quantum-sensor ensembles receiving binary phase-shift keying (BPSK)~\cite{Watkins2026}. Here the design question is what to do \emph{after} acquisition: given the same pilot counts, when should the receiver physically rotate the sensing axis, keep its current setting, or use separately readable quadratures to track the carrier in software?

We connect that decision to the fluorescence counts of a resonantly driven nitrogen-vacancy spin. For a single payload spin, the pilot posterior determines exactly how much additional delay a phase update can tolerate. Its circular mean direction also yields an implementable update rule when estimated from orthogonal pilots. For an ensemble, optical distinguishability and phase-reference ambiguity bound the benefit, and fluctuations of the common carrier contribute a count-variance term that grows quadratically with spin number. Figure~\ref{fig:architecture} shows the control path. Packet simulations then compare its reliability and elapsed time with receivers that have access to parallel quadrature channels.

\section{System Model and Design Objective}\label{sec:model}
\subsection{Microwave signal and sensing sequence}
Consider a microwave source transmitting two known pilots followed by a packet of $L$ unknown BPSK symbols $b_k\in\{-1,+1\}$. A bias field selects an NV spin transition near resonance with the incident microwave magnetic field. The received carrier has an unknown phase $\phi(t)$ relative to the microwave oscillator that programs the sensor. The receiver can set the phase $\alpha$ of its spin preparation and analysis, select an interrogation duration $T$, and allocate $N$ separable spin preparations to an observation. These are distinct controls: $\phi(t)$ is a property of the incoming signal, whereas $\alpha$ is chosen by the receiver. The physical output of each observation is a fluorescence count $C$, not a direct measurement of the carrier phase.

During a measurement cycle, the spins are prepared, interrogated, and optically read out. Let $t_{\rm prep}$ and $t_{\rm ro}$ denote preparation and fluorescence-readout durations, respectively, and let $\tau_j$ be the \emph{baseline} start-to-start spacing between consecutive interrogations. This spacing accommodates readout, the following preparation, fixed idle time, and preprogrammed changes of control phase, so $\tau_j\ge T_j+t_{\rm ro}+t_{\rm prep}$. A measurement-dependent choice may additionally require digital processing and hardware settling; its incremental latency is $\ell$. The packet resource account is then
\begin{align}\label{eq:budgets}
 N_{\rm used}&=N_1+N_2+\sum_{k=1}^L N_{d,k},\\
 t_{\rm pkt}&=t_{\rm prep}+\sum_{j=1}^{L+1}\tau_j
       +T_{L+2}+t_{\rm ro}+n_{\rm fb}\ell.\nonumber
\end{align}
Here $n_{\rm fb}$ counts \emph{measurement-dependent} reconfiguration actions, not every change of programmed phase. The packet model permits the payload interrogation window to move with the feedback decision. It applies when symbol timing accommodates that shift, for example with a transmitter that starts the payload on a receiver-ready signal and includes any coordination delay in the timing budget. In a fixed-slot link with a common guard interval, both actions incur the same phase aging if interrogation starts at the same instant. Then the common attenuation cancels from the axis comparison, and settling imposes a feasibility deadline. The extra time in~\eqref{eq:budgets} is charged on top of the baseline spacing only when feedback is executed. Spin preparations may be parallel within an epoch, so increasing $N$ need not lengthen interrogation but requires more spins and optical collection. A phase-diverse receiver needs separately controlled and read subensembles, whereas the proposed single-phase controller requires a single common phase setting. Equal spin budgets therefore do not imply identical hardware complexity~\cite{Cambria2025,Ensemble2026}.

The relative phase is modeled as a continuously diffusing process, $d\phi(t)=\sqrt{2D}\,dW_t$, with phase-diffusion rate $D$ and Brownian motion $W_t$. Thus the phase increment over an interval $\delta$ has variance $2D\delta$, and the coherence of its first circular harmonic decays according to
\begin{equation}\label{eq:phase_aging}
 \mathbb E\!\left[e^{i[\phi(t+\delta)-\phi(t)]}\right]=e^{-D\delta}.
\end{equation}
This equation quantifies why a pilot-based phase reference becomes stale while the receiver reads out the spins or waits for a new control setting. It does not assume that the phase is frozen during interrogation.

\subsection{From spin response to photon counts}
We use an effective, resonant two-level NV transition to connect the incident microwave to the observed photon counts. Let $\Omega$ be the field-induced Rabi frequency, $\Gamma_2$ the intrinsic transverse dephasing rate, and $\varphi$ the carrier phase at the start of the interrogation. For a symbol $b\in\{-1,+1\}$, the conditional mean spin-population contrast along the optical-readout axis is
\begin{align}\label{eq:response}
 \mathbb E[r_z(T)\mid b,\varphi]&=b\,g_D(T)\cos(\varphi-\alpha),\\
 g_D(T)&=\frac{\Omega}{\omega}e^{-\gamma T/2}\sin(\omega T),\nonumber\\
 \gamma&=\Gamma_2+D,\qquad \omega^2=\Omega^2-\gamma^2/4.\nonumber
\end{align}
The displayed form applies in the underdamped regime, with its continuous hyperbolic form used otherwise. The derivation in Appendix~\ref{app:spin} includes carrier fluctuations \emph{during} the shot. The factor $g_D(T)$ represents the useful coherent response: very short interrogations accumulate little signal, whereas long interrogations lose contrast through dephasing and phase wander. The two-level approximation neglects detuning, hyperfine dynamics, and imperfect pulse preparation; these are important practical control considerations for NV ensembles~\cite{Poulsen2022}.

Optical detection produces different photon statistics for the two spin projections. Write $P_+(c)=e^{-\mu_+}\mu_+^c/c!$ and $P_-(c)=e^{-\mu_-}\mu_-^c/c!$ for the bright and dark conditional count laws. The single-spin observation model is the mixture
\begin{align}\label{eq:photons}
 p(c\mid b,\varphi,\alpha)&=u_c+b\,g_D(T)v_c\cos(\varphi-\alpha),\\
 u_c&=\tfrac12[P_+(c)+P_-(c)],\nonumber\\
 v_c&=\tfrac12[P_+(c)-P_-(c)].\nonumber
\end{align}
This is generally \emph{not} a single Poisson distribution evaluated at the mean brightness~\cite{Hopper2018,Shields2015}. With $\bar\mu=(\mu_++\mu_-)/2$ and $\Delta\mu=\mu_+-\mu_-$, the mean count of $N$ equivalent spins is
\begin{equation}\label{eq:mean_count}
 \mathbb E[C\mid b,\varphi]=N\!\left[\bar\mu+
 \frac{\Delta\mu}{2}b\,g_D(T)\cos(\varphi-\alpha)\right].
\end{equation}
Equation~\eqref{eq:mean_count} is the principal engineering input--output relation. For a fixed phase reference, choosing $\alpha$ controls the separation between the expected fluorescence counts under the two BPSK hypotheses. The full count distribution matters for reliable decoding, especially when few photons are detected. Appendix~\ref{app:ensemble} gives the finite-ensemble law and distinguishes independent spin-projection noise from phase fluctuations common to all spins.

\subsection{Receiver objective and scope}
The receiver must use the noisy pilots to set an informative measurement phase without spending more time and preparations than the benefit warrants. We characterize reliability through the packet-error probability and bit-error rate (BER), and record elapsed time separately. One representative design problem is
\begin{equation}\label{eq:designgoal}
 \min_{\pi}\Pr_{\pi}(\widehat{\mathbf b}\neq\mathbf b),\quad
 N_{\rm used}\le N_{\rm tot},\quad t_{\rm pkt}\le t_{\max},
\end{equation}
where $\pi$ includes the pilot schedule, programmable control phases, feedback decision, and bit decoder; the probability averages over carrier drift, symbol values, and photon observations. Section~\ref{sec:practical} constructs a simple causal receiver, Section~\ref{sec:theory} provides analytical decision limits, and Section~\ref{sec:results} evaluates feasible configurations numerically. Known pilots are essential in this setup: without a phase reference or other symmetry-breaking information, a uniformly unknown carrier phase leaves the sign of an isolated BPSK symbol ambiguous.

\section{Pilot-Aided Receiver Design}\label{sec:practical}
\subsection{Phase acquisition from orthogonal pilots}
The receiver first observes two known $+1$ pilot symbols using orthogonal control phases $\alpha_1=0$ and $\alpha_2=\pi/2$. Equation~\eqref{eq:mean_count} shows that their mean-centered counts probe the cosine and sine components of the relative carrier phase, respectively. These components can be estimated by normalizing the observed counts:
\begin{equation}\label{eq:iq}
 X_j=\frac{2(C_j-N_j\bar\mu)}{N_j\Delta\mu\,g_D(T_j)},\qquad j=1,2.
\end{equation}
With a fixed phase, $\mathbb E[X_1\mid\phi]=\cos\phi$ and $\mathbb E[X_2\mid\phi]=\sin\phi$. Because the pilots occur at different times, the older projection is discounted using the known first-harmonic decay in~\eqref{eq:phase_aging}. For pilot-start separation $\tau$, we use
\begin{equation}\label{eq:phihat}
 \widehat\phi=\operatorname{atan2}
 \!\left(X_2,e^{-D\tau}X_1\right)\pmod{2\pi}.
\end{equation}
The calculation requires two counts, calibration constants, and an angle operation; it does not maintain a full probability distribution for the unknown phase. It is deliberately approximate, since photon fluctuations and carrier drift prevent the two measurements from describing an identical phase. Calibration of $g_D(T_j)$ and the bright/dark photon means is required before this normalization is meaningful.

\subsection{Measurement reconfiguration and symbol decisions}
Let $\mathcal A$ denote the programmable control-phase menu and $\alpha_0\in\mathcal A$ the preselected setting. The estimate $\widehat\phi$ defines a candidate $\widehat\alpha\in\mathcal A$ that maximizes $|\cos(\widehat\phi-\alpha)|$. A physical update is performed only if its predicted alignment benefit survives the additional phase drift:
\begin{equation}\label{eq:iqgate}
 e^{-D\ell}|\cos(\widehat\phi-\widehat\alpha)|>
 |\cos(\widehat\phi-\alpha_0)|+m,
\end{equation}
where $m\ge0$ is a configurable margin. Otherwise, the original setting is retained. The rule is a computationally light \emph{surrogate} for a decision based on the full photon-count likelihood; it does not guarantee lower error for every noisy pilot realization. At most one setting change is made before the payload. The algorithm therefore keeps feedback hardware simple, although the phase reference continues to age while the packet is received.

For each data symbol, the decoder compares the centered photon count with the sign predicted by the current phase reference:
\begin{equation}\label{eq:threshold}
 \widehat b_k=\operatorname{sign}\!\left[
 (C_k-N_d\bar\mu)\cos(\widehat\phi-\alpha_d)\right].
\end{equation}
Ties are assigned to $+1$. The reference angle also resolves the sign ambiguity when a control phase is represented modulo $\pi$. The count-midpoint test is chosen for implementation simplicity; the optimum photon-likelihood test is used only in the analytical benchmark. In longer packets, additional pilots or tracking would be necessary to contain carrier drift.

\subsection{Alternative receiver implementations}
We compare physical phase reconfiguration with two approaches that avoid it. A preprogrammed phase-diverse receiver measures orthogonal (or several) control phases using separately readable subensembles. A decision-directed (DD) software tracker uses these phase-diverse optical counts to update the carrier reference numerically. With $n_I+n_Q=N_d$, its complex count statistic is
\begin{equation}
 Z_k=\frac{C_{I,k}-n_I\bar\mu}{n_I}+
 i\frac{C_{Q,k}-n_Q\bar\mu}{n_Q},
\end{equation}
and the phase recursion is
\begin{align}\label{eq:ddloop}
 \widehat\phi_{k+1}&=\arg\!\left[(1-\kappa)e^{i\widehat\phi_k}
 +\kappa\widehat b_k Z_k/\chi\right],\\
 \chi&=\tfrac12\Delta\mu\,g_D(T_d).\nonumber
\end{align}
The gain $\kappa\in[0,1]$ is calibrated independently of the evaluation data. The tracker avoids measurement-dependent hardware settling if its digital update completes within the fixed inter-epoch slack. It requires independently readable quadrature channels and can suffer phase slips following incorrect symbol decisions. It is a classical \emph{processing method on the same NV observations}, not a conventional antenna receiver.

A third comparator maintains the full circular distribution of the unknown carrier phase and minimizes the predicted bit-error risk for each candidate control. This posterior-guided benchmark is computationally more involved but provides a disciplined reference against which to evaluate the low-complexity estimator. Its discretization and finite action set are described in Appendix~\ref{app:benchmark}. Keeping the practical receiver separate from this benchmark prevents an idealized estimation algorithm from being mistaken for a necessary hardware component.

\section{Performance Analysis and Feedback Limits}\label{sec:theory}
The preceding algorithm uses a convenient pilot phase estimate. To distinguish the value of physical feedback from imperfections of that estimate, we first consider the best decision allowed by a clearly specified single-spin measurement class, and then extend the comparison to the actual photon-count law of a finite ensemble.

\subsection{Exact single-spin feedback threshold}
Let $H=h$ be any record of known-pilot counts and $p(h,\varphi)$ its joint distribution with the carrier phase immediately after acquisition. The unnormalized first circular moment
\begin{equation}\label{eq:AB}
 z_h=\int_0^{2\pi}e^{i\varphi}p(h,\varphi)\,d\varphi
\end{equation}
encodes both the confidence and direction of the available phase reference. Define
\begin{equation}\label{eq:AB2}
 A=\sum_h|z_h|,\qquad
 B=\max_\alpha\sum_h|\Re(e^{-i\alpha}z_h)|.
\end{equation}
The quantity $A$ measures how well an observation-dependent measurement axis can align with the pilots; $B$ describes the best single axis programmed in advance. For the payload shot, define the combined spin and optical discrimination factor
\begin{align}\label{eq:eta}
 \eta_d&=g_D(T_d)\operatorname{TV}(P_+,P_-),\\
 \operatorname{TV}(P_+,P_-)
 &=\tfrac12\sum_c|P_+(c)-P_-(c)|.\nonumber
\end{align}
Here total variation (TV) quantifies how readily two photon-count distributions can be distinguished. We choose an interrogation in the first response lobe, so $g_D(T_d)\ge0$. Let $\Delta$ denote phase aging between pilot acquisition and payload sensing without an update; the update adds delay $\ell$.

\begin{theorem}[Single-spin feedback threshold]\label{th:main}
Both receiver classes may select any sensing axis $\alpha\in[0,\pi)$. Under~\eqref{eq:response}--\eqref{eq:photons}, equiprobable payload bits, any physically admissible known-pilot record, and one photon-resolved payload spin, the minimum bit-error probabilities within the fixed-axis and history-dependent single-axis classes are
\begin{align}\label{eq:risk}
 R_{\rm fixed}&=\frac12\!\left[1-\eta_d e^{-D\Delta}B\right],\\
 R_{\rm feedback}&=\frac12\!\left[1-\eta_d e^{-D(\Delta+\ell)}A\right].\nonumber
\end{align}
For $D>0$, $\eta_d>0$, and $A>0$, forced feedback strictly reduces the error for the same pilot design if and only if
\begin{equation}\label{eq:delay}
 \ell<\frac1D\log\!\left(\frac AB\right).
\end{equation}
Moreover, $2A/\pi\le B\le A$. Hence a forced update cannot improve the error for $\ell\ge D^{-1}\log(\pi/2)$ in this class; the same universal ceiling holds when both policies may optimize over the same feasible pilot set and use the same payload resources.
\end{theorem}

The result formalizes a simple engineering balance. Feedback can align the sensor more accurately with the phase reference obtained from pilots, but that advantage must survive the additional waiting time. For a single payload spin, the maximum reduction in bit error obeys
\begin{equation}\label{eq:ceiling}
 [R_{\rm fixed}-R_{\rm feedback}]_+
 \le\frac{\eta_d e^{-D\Delta}}2
 \left[e^{-D\ell}-\frac2\pi\right]_+.
\end{equation}
Thus low optical discrimination or short spin coherence constrains the \emph{absolute} benefit even when a phase update improves the alignment. Appendix~\ref{app:proof} proves the theorem and its bound.

The same analysis also provides a precise optional-update reference. If the receiver starts with $\alpha_0$ and the observed pilot history has probability $p_h=\Pr(H=h)>0$, the reduction in conditional BER obtainable by reconfiguring to the best single-spin axis is
\begin{equation}\label{eq:priced}
 \Delta P_e(h)=\frac{\eta_de^{-D\Delta}}{2p_h}
 \left[e^{-D\ell}|z_h|-|\Re(e^{-i\alpha_0}z_h)|\right].
\end{equation}
With ties resolved by retaining the default, an update is selected when $\Delta P_e(h)>\lambda\ell$. This comparison also explains the practical phase-menu rule. Write $\theta_h=\arg z_h$ for $|z_h|>0$. With $\lambda=0$ and $\eta_d>0$, the factor $|z_h|/p_h$ cancels from the two conditional risks, giving
\begin{equation}\label{eq:menugate}
 e^{-D\ell}\max_{\alpha\in\mathcal A}|\cos(\theta_h-\alpha)|
 >|\cos(\theta_h-\alpha_0)|.
\end{equation}
Thus~\eqref{eq:iqgate} with $m=0$ implements the exact single-spin axis decision if the pilot estimate equals the posterior direction modulo $\pi$. A zero first moment gives equal risks for all axes. With a positive time price the magnitude matters again, since a weak phase reference yields little absolute error reduction. For ensembles, the relevant quantity is the separation between the full count laws.

\subsection{Finite-ensemble photon discrimination}
With multiple payload spins, the one-spin factorization in~\eqref{eq:risk} is generally lost. The outcome can be a total photon count or a vector of counts from independently controlled subensembles, and common phase fluctuations may correlate their observations. A receiver-level decision can nevertheless be defined directly in terms of the two predicted count laws. Given pilot history $h$, an admissible payload configuration $a\in\mathcal A$, and an additional waiting time $t$, let $Q_h^{\pm}(c\mid a,t)$ denote the count distributions under the two possible payload bits. Their separation is
\begin{equation}\label{eq:finiteTV}
 V_h(a,t)=\frac12\sum_c\left|
 Q_h^+(c\mid a,t)-Q_h^-(c\mid a,t)\right|.
\end{equation}
This quantity has a direct interpretation: the minimum equal-prior binary decision error from the observed counts is $(1-V_h)/2$.

\begin{proposition}[Finite-ensemble feedback comparison]\label{prop:finite}
Consider $N$ independently prepared spins whose optical readouts are conditionally independent given the common carrier trajectory, payload bit, and sensing configuration. For a known-pilot history $H$ and an admissible configuration menu $\mathcal A$, the best fixed and forced-update single-block errors with optimal photon-count decoding are
\begin{align}\label{eq:finiteDecision}
 R_{\rm fix}&=\frac12\left[1-\max_{a\in\mathcal A}
                         \mathbb E_H V_H(a,0)\right],\\
 R_{\rm upd}&=\frac12\left[1-\mathbb E_H
                         \max_{a\in\mathcal A}V_H(a,\ell)\right].\nonumber
\end{align}
For a fixed default configuration $a_0$, updating after history $h$ strictly improves $P_e+\lambda t$ over retaining the default exactly when
\begin{equation}\label{eq:finiteGate}
 \max_{a\in\mathcal A}V_h(a,\ell)-V_h(a_0,0)>2\lambda\ell.
\end{equation}
Let $v=\operatorname{TV}(P_+,P_-)$, and let $q_h$ be the pilot-conditioned phase density before the optional delay. Then any $N$-spin configuration satisfies
\begin{align}\label{eq:finiteBound}
 V_h(a,\ell)&\le\min\{1-(1-v)^N,\;\tau_h(\ell)\},\\
 \tau_h(\ell)&=\operatorname{TV}\!\left(q_h*K_\ell,
                          q_h^\pi*K_\ell\right),\nonumber
\end{align}
where $K_\ell$ is the circular Wiener-phase kernel and $q_h^\pi(\phi)=q_h(\phi-\pi)$.
\end{proposition}

The two limits in~\eqref{eq:finiteBound} represent different bottlenecks. Even a perfectly known phase cannot overcome nearly identical bright/dark photon distributions; conversely, excellent optical counts cannot resolve a BPSK sign when the pilot reference is indistinguishable from its phase-shifted counterpart. The optical bound is deliberately conservative for large $N$, but applies without replacing shared phase noise by independent samples. The exact comparison in~\eqref{eq:finiteDecision} requires the actual predictive photon laws and does not assert that the low-complexity receiver reaches the optimum. Appendix~\ref{app:ensemble} provides the proof.

\subsection{Correlations induced by a shared carrier}
An ensemble samples a common microwave waveform, so averaging the carrier trajectory separately for each spin can understate the uncertainty of the total photon count. Let $r_{\rm path}$ denote the spin-population contrast conditioned on one realized carrier trajectory, with mean $\bar r$ and variance $\sigma_r^2$ across trajectories at a fixed initial phase. Conditional on that path, spin-projection and photon-emission events are independent. The resulting exact variance is
\begin{align}\label{eq:sharedvariance}
 \operatorname{Var}(C\mid\phi_0)
 ={}&N\left(\bar\mu+\tfrac{\Delta\mu}{2}\bar r\right)\\[-2pt]
 &+\frac{N(\Delta\mu)^2}{4}(1-\bar r^2)\nonumber\\[-2pt]
 &+\frac{N(N-1)(\Delta\mu)^2}{4}\sigma_r^2.\nonumber
\end{align}
The three terms arise from fluorescence shot noise, spin-projection fluctuations, and the carrier trajectory shared across spins, respectively. The third scales quadratically with ensemble size before normalization; therefore, large ensembles do not automatically average away the error introduced by a common drifting phase. The covariance arises from the shared classical drive. Appendix~\ref{app:ensemble} derives the identity. Section~\ref{sec:results} measures the resulting discrepancy in the photon-count law and evaluates packet reliability with a shared carrier trajectory.

\FloatBarrier
\section{Numerical Results and Engineering Assessment}\label{sec:results}
\subsection{Simulation setting and validity}
Table~\ref{tab:parameters} gives the operating point for the shared-phase packet simulations. We use an effective resonant coupling $\Omega=\gamma_e B_{\rm eff}$ with $\gamma_e=2\pi\times28$ GHz/T. The packet controller searches twelve equally spaced axes modulo $\pi$. The single-symbol posterior calculation uses six axes and a circular grid of 96 points. The independent packet simulator integrates a common carrier trajectory with 24 substeps per interrogation. Bright and dark photon means are illustrative. Increasing detected photons per spin can require different collection optics or a longer readout~\cite{Hopper2018,Kim2023,Ensemble2026}. The resource sweep instead increases the number of parallel spins at fixed per-spin brightness.

\begin{table}[t]\centering\footnotesize
\caption{Numerical operating point and resource accounting (illustrative)}\label{tab:parameters}
\begin{tabular}{@{}p{.56\columnwidth}p{.35\columnwidth}@{}}\toprule
Parameter & Value\\\midrule
Effective field; phase diffusion & $80$ nT; $160$ rad$^2$/s\\
Intrinsic transverse coherence & $250\,\mu$s\\
Pilot/data interrogation & $90/90\,\mu$s\\
Prep/readout; feedback delay & $10/20;\ 80\,\mu$s\\
Mean photons, bright/dark & $0.45/0.15$ per spin\\
Packet size and total spin uses & $8$ bits; $432$ preparations\\
Phase menu; reference grid & $12$; $6$/$96$ (reference)\\\bottomrule
\end{tabular}
\end{table}

In the packet simulator, the baseline interrogation-start spacing is $150\,\mu$s, comprising $90\,\mu$s interrogation, $20\,\mu$s optical readout, $10\,\mu$s preparation for the next shot, and $30\,\mu$s for preprogrammed controls and fixed slack. The first preparation is charged separately; the last readout ends the packet. Software-tracker calculations are assumed to fit within this slack. Only a measurement-conditioned global phase update adds $80\,\mu$s, so the base eight-bit packet duration is $1.47$ ms.\par

\subsection{Single-symbol reliability and feedback delay}
The spin response~\eqref{eq:response} is maximized at an intermediate
interrogation time, because longer exposure eventually loses coherence.
Figures~\ref{fig:exact} and~\ref{fig:latency} compare reduced-model control
choices with separately optimized fixed and feedback phase menus. They evaluate finite ensembles using the reduced likelihood. The crossings in the delay curves illustrate the theorem's central mechanism: a more informative measurement can become inferior after its reference has aged.
\begin{figure}[t]\centering
\includegraphics[width=.97\columnwidth]{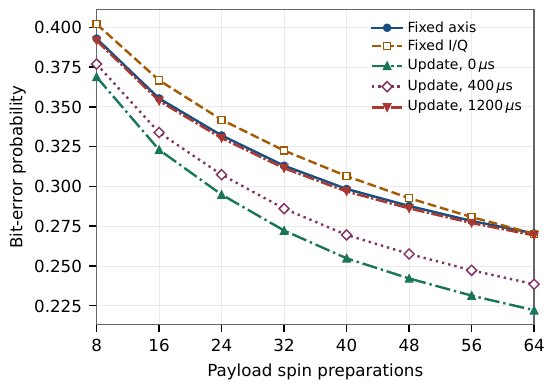}
\caption{Reduced-model BER versus payload spin resources: fixed single-phase,
fixed two-quadrature, and feedback reception at three distinct delays.
Each pilot uses 24 spin preparations.}
\label{fig:exact}
\end{figure}
\begin{figure}[t]\centering
\includegraphics[width=.97\columnwidth]{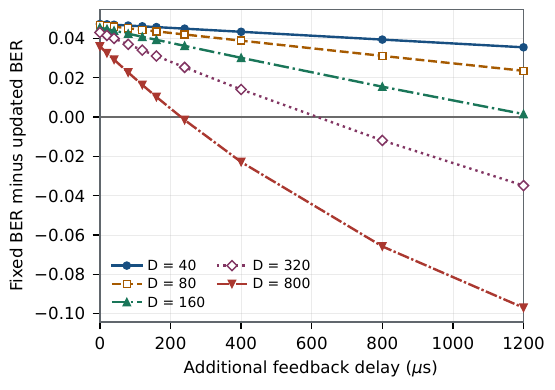}
\caption{Fixed-axis BER minus forced-update BER versus feedback delay for five
carrier-diffusion rates; positive values favor reconfiguration.}
\label{fig:latency}
\end{figure}

\subsection{Allocation of spins and optical readout quality}
\begin{figure}[t]\centering
\includegraphics[width=.97\columnwidth]{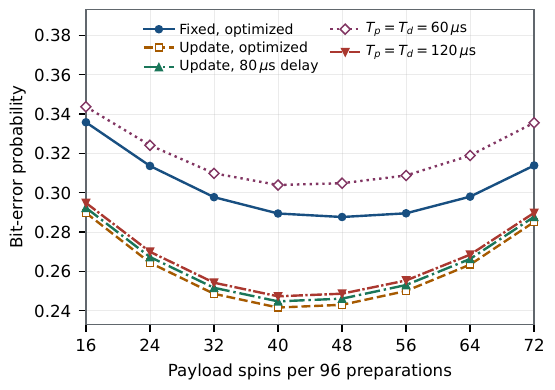}
\caption{Finite-menu optimization under a common total of 96 spin preparations, comparing fixed and adaptive strategies with free or prescribed interrogation times and indicated feedback delays.}\label{fig:resources}
\end{figure}
\begin{figure}[t]\centering
\includegraphics[width=.97\columnwidth]{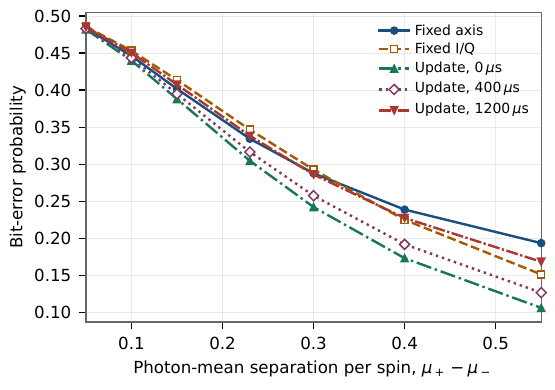}
\caption{Sensitivity of reduced-model receiver error to detectable bright--dark photon separation, with five control/latency settings. Optical-readout quality constrains all policies.}\label{fig:contrast}
\end{figure}
The allocation study in Fig.~\ref{fig:resources} shows that the number of pilot spins and payload spins must be designed together; more pilots need not yield lower BER when they displace data measurements. Figure~\ref{fig:contrast} demonstrates that better control cannot compensate indefinitely for poor photon discrimination. In the tested discrete design menu, the preferred interrogation duration is near $90\,\mu$s, rather than simply the longest available shot.

\subsection{Packet-level receiver comparison}
We next evaluate six implementable strategies under exactly the same
$N_1=24$, $N_2=24$, and $N_d=48$ spin preparations at each of eight data
symbols: (i) a fixed one-phase receiver; (ii) preprogrammed orthogonal
subensembles; (iii) four parallel phase subensembles; (iv) the conventional
DD software tracker~\eqref{eq:ddloop}; and (v)--(vi) the proposed pilot-I/Q
receiver with unconditional or optional one-time reconfiguration. All use
the same two known pilots and a shared Wiener carrier path. The phase-diverse
approaches require separately accessible optical subensembles, whereas the
proposed receiver requires a physical phase-setting update. The fixed receiver uses the same pilots to resolve its sign.

The decision-directed (DD) loop gain is selected from a prespecified seven-value menu on an independent calibration set (1600 packets at smaller ensembles, 1000 at the two largest ensembles); the I/Q split is screened at 25\%, 50\%, and 75\% for small ensembles and fixed to 50\% at larger sizes. The fixed I/Q and four-phase receivers use their declared balanced splits. Each quadrature statistic is normalized by the spin count \emph{of that subensemble} at the evaluated resource level, consistent with~\eqref{eq:ddloop}. A separate set of 6000 packets per policy and operating point provides held-out evaluation. The proposed
pilot-I/Q gate has no fitted margin ($m=0$); no test outcomes are used to
select its control. The simulation propagates the \emph{same carrier phase}
through all spins within an interrogation, rather than drawing independent
phase noise for each spin. Exact one-symbol posterior results remain the
analytical benchmark in Section~\ref{sec:theory}; multi-symbol packet results
here are held-out performance measurements from the stated simulator.

\begin{figure}[!t]\centering
\includegraphics[width=.98\columnwidth]{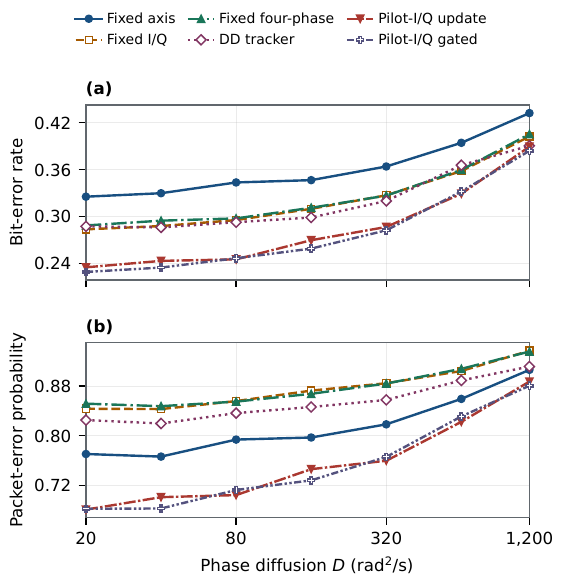}
\caption{Eight-bit packet BER (a) and packet-error probability (b) versus carrier-phase diffusion, at the base spin budget. Only the DD tracker is tuned on calibration packets.}
\label{fig:packet}
\end{figure}
At the nominal diffusion level, the gated pilot-I/Q receiver has BER
$0.2590\pm0.0069$ (approximate 95\% half-width) and packet-error probability
$0.7283\pm0.0112$; the independently calibrated DD software tracker obtains
BER $0.2988$ and packet error $0.8462$. These estimates use independent packets as the sampling unit for both metrics. The substantial packet error confirms that the base spin budget is insufficient for reliable uncoded packet reception. The useful engineering question is therefore how rapidly reliability improves with sensing resources, and whether extra hardware control delay is justified relative to an independently addressed dual-quadrature software receiver.

\begin{figure}[t]\centering
\includegraphics[width=.98\columnwidth]{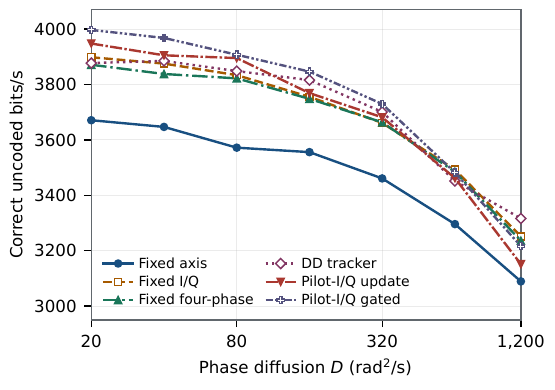}
\caption{Correct uncoded bit decisions per elapsed second for the same six
policies; the time account includes all measurement cycles and physical
reconfiguration. The rate is the total number of correctly decoded bits divided by the total simulated packet time.}
\label{fig:rate}
\end{figure}
Figure~\ref{fig:rate} exposes the principal practical cost of adaptation:
better BER need not imply better time-normalized reception, because an NV
control update consumes real elapsed time. Software phase tracking is free
of this particular \emph{hardware-settling} cost, but it requires parallel
quadrature controls and still suffers decision errors and phase slips.

\subsection{Reliability versus spin-preparation resources}
To isolate the effect of sensing resources without assuming a stronger optical readout, we keep the per-spin bright/dark photon means fixed and scale
the number of parallel preparations at every pilot and payload epoch. This
assumes a uniform driving field and sufficient optical collection for the
larger ensemble; it is an architectural sensitivity study, not a calibrated
experimental forecast. The elapsed interrogation and readout intervals are
held fixed, while the spin-preparation budget increases proportionally.
\begin{figure}[!t]\centering
\includegraphics[width=.98\columnwidth]{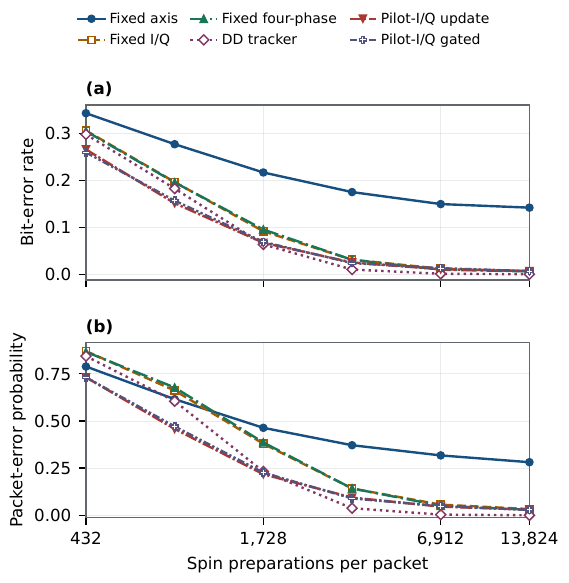}
\caption{Eight-bit packet BER (a) and packet-error probability (b) versus spin-preparation budget, with per-spin optical statistics held fixed.}
\label{fig:operating}
\end{figure}
At $3456$ spin preparations per packet ($192$ per pilot and $384$ per payload symbol), the feedback-gated receiver attains packet error $0.0893\pm0.0072$ (approximate 95\% half-width), versus $0.1422\pm0.0088$ for fixed orthogonal-I/Q reception. A separately calibrated dual-quadrature DD tracker achieves lower packet error, $0.0375\pm0.0048$, and BER $0.0101$ under the same spin budget, without measurement-dependent hardware settling. At this point gating skips $10.7\%$ of hardware updates. Forced updating gives packet error $0.0942$; its difference from $0.0893$ lies within Monte Carlo uncertainty (approximate 95\% interval for forced minus gated, $[-0.0055,0.0152]$). The measurable gain from skipping updates here is shorter mean packet time. At $6912$ preparations the DD policy attains $0.0037$ packet error, versus $0.0487$ for gated hardware feedback. At $13\,824$ preparations the DD receiver records no packet errors in 6000 evaluation packets; this is \emph{not} a zero-error guarantee (a rough one-sided 95\% upper bound is $3/6000$). Because carrier-phase uncertainty is shared across successive decisions, packet-error probability can differ from the independent-bit expression $1-(1-\mathrm{BER})^8$. Those large-ensemble results require independently accessible quadratures, adequate optical throughput, and fast enough digital phase tracking. By contrast, the single-phase adaptive design avoids that parallel-readout requirement but pays extra latency. 

A second timing test holds the payload onset fixed across all receivers. Each packet reserves the same $80\,\mu$s guard after acquisition, giving a common duration of $1.55$ ms. The single-channel receiver uses the guard to apply its pilot-selected axis, and the DD receiver retains the gains and splits calibrated for the original schedule. At 432 preparations, packet-error probabilities over 6000 fresh packets per receiver are $0.742$ for the updated axis, $0.844$ for DD, and $0.870$ for fixed I/Q. At 3456 preparations they are $0.094$, $0.035$, and $0.157$, respectively. The ordering between the updated axis and software tracking therefore reverses with spin budget under both schedules. A common guard removes the elapsed-time saving from skipping the update.

\subsection{Validation of the ensemble approximation}
\begin{figure}[t]\centering
\includegraphics[width=.97\columnwidth]{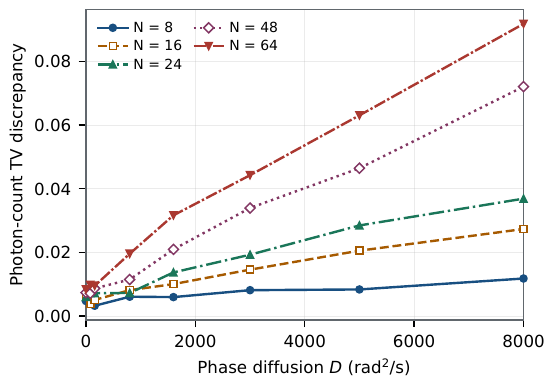}
\caption{Discrepancy between the average-spin ensemble likelihood and a common stochastic carrier path for five ensemble sizes. The nonzero zero-diffusion floor includes simulation sampling error; discrepancy grows for large ensembles and fast shared phase fluctuations.}\label{fig:stress}
\end{figure}
The reduced posterior calculation uses $L_N(c\mid\bar r)$ and predicts phase at the next epoch from the start phase. A shared trajectory couples the current count to the terminal phase. The packet simulator retains that coupling. Figure~\ref{fig:stress} checks its distributional accuracy against a separate common-phase simulation. At nominal diffusion $D=160$ rad$^2$/s the empirical photon-count discrepancy is about $0.0087$ for $N=48$; at $D=8000$ rad$^2$/s it rises to $0.0720$ for $N=48$ and $0.0917$ for $N=64$. These observations agree with the covariance mechanism in~\eqref{eq:sharedvariance}. They do not uniformly bound packet BER mismatch, and high-diffusion operation would require a joint spin--carrier-state filter or a calibrated surrogate.

For experimental transfer, the effective microwave coupling, bright/dark photon histograms, spin-coherence time, control-settling delay, and accessibility of phase-diverse subensembles must be measured together. The reported near-resonant, two-level and Poisson readout assumptions do not include detuning, hyperfine structure, optical pumping transients, background drift, or a matched electronic RF noise front end. A device-specific evaluation would use these measured quantities in the count likelihood and the packet schedule.

\FloatBarrier
\section{Conclusion}
A diamond-spin microwave receiver must weigh a better phase-aligned measurement against the real time and coherence consumed by applying it. We developed a pilot-I/Q controller that requires at most one sensor-phase update per packet. An exact single-spin delay result and a finite-ensemble photon-discrimination bound identify complementary limitations from phase-reference aging and optical readout. In independently evaluated shared-phase simulations, hardware feedback improves reliability over fixed phase-diverse control at modest resources, but a separately calibrated dual-quadrature software tracker performs better once sufficient photons are available. Thus the preferred architecture depends not only on error probability and control latency, but also on the availability of independently read spin subensembles and on whether feedback computation can fit within the symbol schedule. For a single optical channel, the pilot direction and effective settling delay determine when an axis change is useful. Independently readable quadratures allow the receiver to spend the same spin budget on a continuously refreshed software reference.

\appendices
\section{Spin Dynamics and the Photon Observation Law}\label{app:spin}
We outline the physical assumptions leading to~\eqref{eq:response}. In the rotating-wave approximation, a resonantly driven effective spin-$1/2$ transition is governed by
\begin{equation}\label{eq:ham}
 \frac{H_b(t)}{\hbar}=\frac{b\Omega}{2}
 [\cos\phi(t)\sigma_x+\sin\phi(t)\sigma_y],
\end{equation}
and independent Markovian transverse dephasing is represented by
\begin{equation}\label{eq:me}
 \dot\rho=-i[H_b/\hbar,\rho]
 +\frac{\Gamma_2}{2}(\sigma_z\rho\sigma_z-\rho).
\end{equation}
The relative carrier phase satisfies $d\phi=\sqrt{2D}\,dW_t$. Let $\mathbf r=(r_x,r_y,r_z)$ be the Bloch vector, and express it in the instantaneous carrier frame as $\mathbf u=R_z(-\phi)\mathbf r$. It\^o's formula supplies both a stochastic rotation $-\sqrt{2D}J_z\mathbf u\,dW_t$ and the corresponding drift $D J_z^2\mathbf u\,dt$, where $J_z\mathbf u=\mathbf e_z\times\mathbf u$ and $J_z^2\mathbf u=(-u_x,-u_y,0)$. After averaging over the phase trajectories, the Bloch components satisfy
\begin{align}\label{eq:blochmean}
 \dot{\bar u}_x&=-(\Gamma_2+D)\bar u_x,\\
 \dot{\bar u}_y&=-(\Gamma_2+D)\bar u_y-b\Omega\bar u_z,\nonumber\\
 \dot{\bar u}_z&=b\Omega\bar u_y.\nonumber
\end{align}
The preparation phase is chosen so that $\bar u_y(0)=\cos(\varphi-\alpha)$ and $\bar u_z(0)=0$. Eliminating $\bar u_y$ from~\eqref{eq:blochmean} yields
\begin{equation}\label{eq:ode}
 \ddot{\bar u}_z+\gamma\dot{\bar u}_z+\Omega^2\bar u_z=0,
 \quad \dot{\bar u}_z(0)=b\Omega\cos(\varphi-\alpha),
\end{equation}
with $\gamma=\Gamma_2+D$. Its solution is exactly~\eqref{eq:response}. For $\gamma>2\Omega$, replace $\sin(\omega T)/\omega$ by $\sinh(\kappa T)/\kappa$, with $\kappa^2=\gamma^2/4-\Omega^2$; at critical damping, the factor tends to $T$. The derivation establishes the \emph{conditional mean} spin response. It does not assert statistical independence of ensemble spins after their common carrier trajectory has been marginalized.

For optical readout, the postinterrogation spin projection is bright with probability $(1+r_z)/2$ and dark otherwise. Mixing the conditional count distributions $P_+$ and $P_-$ gives~\eqref{eq:photons}. This state-to-photon construction is valid for the specified fluorescence model, independent of an artificial Gaussian approximation to the count statistics. Calibration of $\mu_\pm$, the effective field-to-$\Omega$ coupling, and the timing parameters is required before extrapolation to a particular NV device.

\section{Proof of the Single-Spin Feedback Result}\label{app:proof}
Write $p_h=\Pr(H=h)>0$ and $m_h=z_h/p_h$. By~\eqref{eq:phase_aging}, the posterior first harmonic at the data shot is $e^{-D\delta}m_h$ after waiting $\delta$. Marginalizing~\eqref{eq:photons} over that conditional phase law gives, for $b=\pm1$,
\begin{equation}\label{eq:conditionalphoton}
 Q_h^b(c\mid\alpha,\delta)
 =u_c+b\,g_D(T_d)v_c e^{-D\delta}
        \Re(e^{-i\alpha}m_h).
\end{equation}
For two equally probable hypotheses, the minimum count-decision error is one half of one minus the total variation between their distributions. By~\eqref{eq:conditionalphoton} that total variation equals
\begin{equation}\label{eq:singleTV}
 \operatorname{TV}(Q_h^+,Q_h^-)
 =\eta_d e^{-D\delta}|\Re(e^{-i\alpha}m_h)|.
\end{equation}
A fixed sensing axis $\alpha$ is selected before observing $h$; averaging the optimal photon decision error over histories gives $\frac12[1-\eta_d e^{-D\Delta}\sum_h|\Re(e^{-i\alpha}z_h)|]$. Maximizing over the common axis establishes the first line of~\eqref{eq:risk}. When an update is allowed, the sensing axis can be selected separately for each history. Since $\max_\alpha|\Re(e^{-i\alpha}z_h)|=|z_h|$, and the update adds $\ell$, the second line follows. Their direct comparison proves~\eqref{eq:delay}.

For completeness, $|z_h|\le p_h$ and $\sum_hp_h=1$ imply $A\le1$, while $|\Re(e^{-i\alpha}z_h)|\le|z_h|$ implies $B\le A$. Uniformly averaging any fixed projection over $\alpha\in[0,\pi)$ gives
\begin{equation}
 \frac1\pi\int_0^\pi\sum_h
 |\Re(e^{-i\alpha}z_h)|\,d\alpha=\frac{2A}{\pi}.
\end{equation}
Therefore the maximum projection obeys $B\ge2A/\pi$. Combining the inequalities proves the universal delay ceiling and~\eqref{eq:ceiling}. The ceiling remains valid when the two policies independently optimize pilots over the same feasible set: for the pilots selected by the updated design, the best fixed receiver cannot perform worse than the fixed receiver evaluated on that same pilot design. This comparison requires identical payload resources.

For an optional update, subtract the two \emph{conditional} risks after observing $h$ rather than comparing pilot-averaged risks. Equation~\eqref{eq:singleTV} yields precisely~\eqref{eq:priced}. An update adds the latency penalty $\lambda\ell$, giving the stated optimal condition. When $\mu_+>\mu_->0$, the single-spin optical likelihood ratio is increasing in the observed count, with threshold
\begin{equation}
 k_* =\left\lceil\frac{\mu_+-\mu_-}{\log(\mu_+/\mu_-)}\right\rceil.
\end{equation}
Consequently, $\operatorname{TV}(P_+,P_-)=\Pr_{\mu_+}(C\ge k_*)-\Pr_{\mu_-}(C\ge k_*)$, relating the readout factor in~\eqref{eq:eta} directly to a photon-counter calibration.\hfill$\blacksquare$

\section{Finite-Ensemble Count Laws and Proof}\label{app:ensemble}
First suppose $N$ independently prepared spins experience a given deterministic trajectory with resulting contrast $r$. The number $j$ projected into the bright state is binomial with parameter $(1+r)/2$. Conditional on $j$, the total fluorescence count is Poisson with mean $j\mu_++(N-j)\mu_-$. Therefore
\begin{align}\label{eq:block}
 L_N(c\mid r)&=\sum_{j=0}^{N}\binom Nj
 \left(\frac{1+r}{2}\right)^j\left(\frac{1-r}{2}\right)^{N-j}\\[-2pt]
 &\qquad\times\operatorname{Poi}\!\left(c;j\mu_++(N-j)\mu_-\right).\nonumber
\end{align}
When the spins experience the same random phase trajectory, the proper conditional law is $\mathbb E_{\rm path}[L_N(c\mid r_{\rm path})]$. In general it differs from $L_N(c\mid\mathbb E_{\rm path}r_{\rm path})$; equality of mean responses does not imply equality of count distributions.

The variance identity~\eqref{eq:sharedvariance} follows explicitly from total variance. Conditional on a realized trajectory, a single spin produces a count with mean $m(r)=\bar\mu+\Delta\mu r/2$ and variance
\begin{equation}
 s^2(r)=m(r)+\frac{(\Delta\mu)^2}{4}(1-r^2).
\end{equation}
For $N$ conditionally independent spins, $\mathbb E(C\mid r)=Nm(r)$ and $\operatorname{Var}(C\mid r)=Ns^2(r)$. Averaging over the shared trajectory gives
\begin{align}
 \operatorname{Var}(C\mid\phi_0)
 &=N\mathbb E[s^2(r_{\rm path})]
   +N^2\operatorname{Var}[m(r_{\rm path})]\\
 &=N\left(\bar\mu+\frac{\Delta\mu}{2}\bar r\right)
   +\frac{N(\Delta\mu)^2}{4}(1-\bar r^2)\nonumber\\
 &\quad+\frac{N(N-1)(\Delta\mu)^2}{4}\sigma_r^2,\nonumber
\end{align}
which also shows why the pilot statistics in~\eqref{eq:iq} can have extra uncertainty beyond a naive independent-spin approximation.

\noindent\textit{Proof of Proposition~\ref{prop:finite}.}
For each history $h$ and configuration $a$, the two equiprobable bit hypotheses yield predicted photon-count laws $Q_h^\pm$. The minimum error of any count-based bit decision is $(1-V_h(a,t))/2$. A fixed configuration must be chosen before $H$ is observed, so maximizing the averaged separation produces the first expression in~\eqref{eq:finiteDecision}. An updated configuration can be optimized for each $h$ but experiences the added delay $\ell$, producing the second expression. Including a cost $\lambda\ell$ per update and comparing the two conditional risks yields~\eqref{eq:finiteGate}.

For the optical limit, write $v=\operatorname{TV}(P_+,P_-)$. At any fixed carrier trajectory, the two bit-dependent single-spin photon laws are convex mixtures of $P_+$ and $P_-$ with opposite spin contrasts, so their total variation cannot exceed $v$. The total variation of $N$ conditionally independent count observations is at most $1-(1-v)^N$: this follows by coupling each pair of single-spin observations to agree with probability at least $1-v$ and taking the product coupling. Marginalizing over the common trajectory and grouping separately resolved counts into an aggregate record are stochastic transformations and cannot increase total variation. This proves the first bound in~\eqref{eq:finiteBound}.

For the phase-reference limit, changing the BPSK bit sign is equivalent to shifting the carrier phase by $\pi$ in~\eqref{eq:ham}. Conditional on the pilots, the two corresponding phase laws at payload time are $q_h*K_\ell$ and $q_h^\pi*K_\ell$. Both pass through the same physical spin-control and photon-count channel, which cannot increase their total variation. This establishes the second bound and completes the proof. The reasoning presumes that the phase-noise and photon-emission mechanisms themselves do not depend on the bit hypothesis, apart from the specified microwave drive.\hfill$\blacksquare$

\section{Posterior-Guided Reference Receiver}\label{app:benchmark}
The reference receiver represents the carrier phase on a circular grid. Its predicted distribution is propagated with the wrapped-Gaussian kernel, whose $n$th circular harmonic decays as $e^{-Dn^2\delta}$. A known pilot count updates the distribution through the count likelihood~\eqref{eq:block}. For unknown payload bits, the likelihood is averaged over both possible bits; a receiver never uses the transmitted bit as an observation. If $W_h(\varphi)$ is the normalized predicted phase weight and $L_{b,c}^{a,t}(\varphi)$ is the count law for a candidate action, its conditional error is
\begin{equation}
 \mathcal R_h(a,t)=\frac12\sum_c\min_{b=\pm1}
 \int W_h(\varphi)L_{b,c}^{a,t}(\varphi)\,d\varphi.
\end{equation}
The fixed reference minimizes the pilot-averaged risk over the phase menu. The updated reference minimizes the conditional risk over that menu after accounting for its delay. All likelihoods in this calculation use the reduced model.

\balance
\bibliographystyle{IEEEtran}
\bibliography{references}

@article{Degen2017,author={C. L. Degen and F. Reinhard and P. Cappellaro},title={Quantum sensing},journal={Reviews of Modern Physics},volume={89},pages={035002},year={2017},doi={10.1103/RevModPhys.89.035002}}

@article{Doherty2013,author={M. W. Doherty and N. B. Manson and P. Delaney and F. Jelezko and J. Wrachtrup and L. C. L. Hollenberg},title={The nitrogen-vacancy colour centre in diamond},journal={Physics Reports},volume={528},pages={1--45},year={2013},doi={10.1016/j.physrep.2013.02.001}}

@article{Barry2020,author={J. F. Barry and J. M. Schloss and E. Bauch and M. J. Turner and C. A. Hart and L. M. Pham and R. L. Walsworth},title={Sensitivity optimization for {NV}-diamond magnetometry},journal={Reviews of Modern Physics},volume={92},pages={015004},year={2020},doi={10.1103/RevModPhys.92.015004}}

@article{SolidSpins2023,author={X.-D. Chen and E.-H. Wang and L.-K. Shan and S.-C. Zhang and C. Feng and Y. Zheng and Y. Dong and G.-C. Guo and F.-W. Sun},title={Quantum enhanced radio detection and ranging with solid spins},journal={Nature Communications},volume={14},pages={1288},year={2023},doi={10.1038/s41467-023-36929-8}}

@article{Zeng2025,author={Q. Zeng and J. Zhang and M. Gupta and Z. Chu and K. Huang},title={Towards ubiquitous radio access using nanodiamond based quantum receivers},journal={Communications Engineering},volume={4},pages={60},year={2025},doi={10.1038/s44172-025-00396-4}}

@misc{Barson2026,author={H. A. M. Taylor and J. P. Duff and L. B. Miller and A. Tritt and M. S. Oldfield and K. Helmerson and M. S. J. Barson},title={High-bandwidth heterodyne sensing with bulk {NV} diamond},howpublished={Physical Review Applied, accepted for publication},year={2026},url={https://doi.org/10.1103/tf8y-nlsk}}

@misc{Watkins2026,author={W. M. Watkins and L. Norris and P. Titum},title={Performance limits of a quantum receiver for detecting phase-modulated communication signals},howpublished={arXiv:2602.07123},year={2026},url={https://arxiv.org/abs/2602.07123}}

@article{McMichael2025,author={M. Kelley and R. D. McMichael},title={Weighing unequal parameter importance and measurement expense in adaptive quantum sensing},journal={Journal of Applied Physics},volume={137},pages={074401},year={2025},doi={10.1063/5.0251881}}

@article{Poulsen2022,author={A. F. L. Poulsen and J. D. Clement and J. L. Webb and R. H. Jensen and L. Troise and K. Berg-S{\o}rensen and A. Huck and U. L. Andersen},title={Optimal control of a nitrogen-vacancy spin ensemble in diamond for sensing in the pulsed domain},journal={Physical Review B},volume={106},pages={014202},year={2022},doi={10.1103/PhysRevB.106.014202}}

@article{Jauch2026,author={I. Jauch and P. Tighineanu and P. Tritschler and T. Strohm and T. Fuchs and F. Jelezko},title={Adaptive and Robust Control of Diamond Quantum Sensors via Meta-Learning},journal={Advanced Quantum Technologies},volume={9},pages={e70385},year={2026},doi={10.1002/qute.70385}}

@article{BerryWiseman2002,author={D. W. Berry and H. M. Wiseman},title={Adaptive quantum measurements of a continuously varying phase},journal={Physical Review A},volume={65},pages={043803},year={2002},doi={10.1103/PhysRevA.65.043803}}

@article{BerryWiseman2004,author={D. T. Pope and H. M. Wiseman and N. K. Langford},title={Adaptive phase estimation is more accurate than nonadaptive phase estimation for continuous beams of light},journal={Physical Review A},volume={70},pages={043812},year={2004},doi={10.1103/PhysRevA.70.043812}}

@article{Hopper2018,author={D. A. Hopper and H. J. Shulevitz and L. C. Bassett},title={Spin Readout Techniques of the Nitrogen-Vacancy Center in Diamond},journal={Micromachines},volume={9},pages={437},year={2018},doi={10.3390/mi9090437}}

@article{Shields2015,author={B. J. Shields and Q. P. Unterreithmeier and N. P. de Leon and H. Park and M. D. Lukin},title={Efficient Readout of a Single Spin State in Diamond via Spin-to-Charge Conversion},journal={Physical Review Letters},volume={114},pages={136402},year={2015},doi={10.1103/PhysRevLett.114.136402}}

@article{Cambria2025,author={M. Cambria and S. Chand and C. M. Reiter and S. Kolkowitz},title={Scalable Parallel Measurement of Individual Nitrogen-Vacancy Centers},journal={Physical Review X},volume={15},pages={031015},year={2025},doi={10.1103/jdzq-jbfz}}

@article{Ensemble2026,author={R. Maier and C.-I. Ho and A. Denisenko and M. Davydova and P. Knittel and J. Wrachtrup and V. Vorobyov},title={Readout of a solid state spin ensemble at the projection noise limit},journal={Nature Communications},volume={17},pages={4028},year={2026},doi={10.1038/s41467-026-72721-0}}

@article{Kim2023,author={L. Kim and H. Choi and M. E. Trusheim and H. Wang and D. R. Englund},title={Nanophotonic quantum sensing with engineered spin-optic coupling},journal={Nanophotonics},volume={12},number={3},pages={441--449},year={2023},doi={10.1515/nanoph-2022-0682}}
\makeatletter\accesseditorstrue\makeatother
\EOD
\end{document}